\documentclass[technote,10pt]{IEEEtran}

\IEEEoverridecommandlockouts
\usepackage{cite}
\usepackage{amsmath,amsfonts,amssymb,amsbsy,amstext,amsthm}
\usepackage[ruled]{algorithm}
\usepackage{algorithmic}
\usepackage{graphicx}
\usepackage{textcomp}
\usepackage{bm}
\usepackage{xcolor}
\usepackage[nolist,printonlyused]{acronym}
\usepackage{hyperref}
\usepackage{dsfont}
\begin{acronym}
\acro{TV}{television}
\acro{CPU}{central processing unit}
\acro{CPRI}{common public radio interface}
\acro{PLF}{predecessor leader following}
\acro{WHO}{world health organization}
\acro{PL}{platoon leader}
\acro{LTE-V2V}{long term evolution V2V}
\acro{PM}{platoon member}
\acro{CAM}{cooperative awareness message}
\acro{D2D}{device-to-device}
\acro{TVWS}{\acs*{TV} white spaces}
\acro{RSRQ}{reference signal received quality}
\acro{RSSI}{received signal strength indicator}
\acro{SWOT}{strengths, weaknesses, opportunities, and threats}
\acro{AP}{access point}
\acro{MS}{mobile station}
\acro{ID}{identity}
\acro{NCL}{neighbor cell list}
\acro{WLAN}{wireless local area network}
\acro{WiMAX}{worldwide interoperability for microwave access}
\acro{LTE}{long term evolution}
\acro{DMRS}{demodulation reference signal}
\acro{UL}{uplink}
\acro{DL}{downlink}
\acro{DPC}{dirty paper coding}
\acro{TTI}{transmission time interval}
\acro{CSI}{channel state information}
\acro{CoMP}{coordinated multi-point}
\acro{ISM}{industrial, scientific and medical}
\acro{3GPP}{3rd generation partnership project}
\acro{M2M}{machine-to-machine}
\acro{MTC}{machine type communication}
\acro{eNB}{evolved node B}
\acro{PRB}{physical resource block}
\acro{PF}{proportional fair}
\acro{RRA}{radio resource allocation}
\acro{MCS}{modulation and coding scheme}
\acro{SCM}{spatial channel model}
\acro{SNR}{signal-to-noise ratio}
\acro{PDF}{probability density function}
\acro{GPS}{global positioning system}
\acro{P2P}{peer to peer}
\acro{MAC}{medium access control}
\acro{OFDM}{orthogonal frequency division multiplexing}
\acro{CCL}{common cell list}
\acro{5G}{fifth generation}
\acro{FDD}{frequency division duplex}
\acro{TDD}{time division duplex}
\acro{EPA}{equal power allocation}
\acro{MR}{maximum ratio}
\acro{PhD}{doctor of philosophy}
\acro{3G}{third generation}
\acro{4G}{fourth generation}
\acro{ANACOM}{autoridade nacional de comunicações}
\acro{AF}{array factor}
\acro{ARPU}{average return per user}
\acro{AHP}{analytic hierarchy process}
\acro{B3G}{beyond \acs*{3G}}
\acro{BSc}{bachelor of science}
\acro{BV}{beamforming vector}
\acro{BWA}{broadband wireless access}
\acro{CAPEX}{capital expenditure}
\acro{COGEU}{cognitive radio systems for efficient sharing of TV white spaces in European context}
\acro{CR}{cognitive radio}
\acro{CV}{combining vector}
\acro{DEnKF}{deterministic ensemble Kalman filter}
\acro{DFT}{discrete Fourier transform}
\acro{DSRC}{dedicated short-range communication}
\acro{DVB-H}{handled digital video broadcast}
\acro{DVB-T}{terrestrial digital video broadcast}
\acro{DTT}{digital terrestrial \acs*{TV}}
\acro{ETSI}{European telecommunications standards institute}
\acro{FCC}{federal communications commission}
\acro{FP7}{seventh framework programme}
\acro{HM}{high mobility}
\acro{HSPA}{high speed packet access}
\acro{IA}{initial access}
\acro{ICT}{information and communication technologies}
\acro{IEEE}{institute of electrical and electronics engineers}
\acro{DySPAN-SC}{\acs*{IEEE} DySPAN standards committee}
\acro{ITU}{international telecommunications union}
\acro{LAN}{local area network}
\acro{LS}{least squares}
\acro{LSFD}{large-scale fading decoding}
\acro{MBMS}{multimedia broadcast multicast services}
\acro{MVNO}{mobile virtual network operator}
\acro{MT}{mobile terminal}
\acro{MF}{matched filter}
\acro{MSc}{master of science}
\acro{NPV}{net present value}
\acro{NRA}{national regulator authority}
\acro{OFCOM}{office of communications}
\acro{OMP}{orthogonal matching pursuit}
\acro{OPEX}{operational expenditure}
\acro{PHY}{physical}
\acro{PMSE}{programme making and special events}
\acro{PT}{Portugal}
\acro{RRS}{reconfigurable radio systems}
\acro{QoE}{quality of experience}
\acro{SAE}{system architecture evolution}
\acro{SDR}{software defined radio}
\acro{SLA}{service level agreement}
\acro{SVD}{singular-value decomposition}
\acro{TETRA}{trans European trunked radio access}
\acro{ULA}{uniform linear array}
\acro{UHF}{ultra high frequency}
\acro{UK}{United Kingdom}
\acro{UMTS}{universal mobile telecommunications system}
\acro{UPA}{uniform planar array}
\acro{UTRA}{\acs*{UMTS} terrestrial radio access}
\acro{USP}{unique selling point}
\acro{USA}{United States of America}
\acro{WG}{working group}
\acro{Wi-Fi}{wireless fidelity}
\acro{WRC}{world radiocommunication conference}
\acro{WSD}{white space device}
\acro{OM}[O\&M]{operational \& maintenance}
\acro{IRR}{internal rate of return}
\acro{EKF}{extended Kalman filter}
\acro{EnKF}{ensemble Kalman filter}
\acro{EPC}{evolved packet core}
\acro{E-UTRAN}{evolved universal terrestrial radio access network}
\acro{MME}{mobility management entity}
\acro{PCRF}{policy and charging rules function}
\acro{S-GW}{service gateway}
\acro{P-GW}{packet gateway}
\acro{EIRP}{equivalent isotropically radiated power}
\acro{CBR}{constant bit rate}
\acro{RRM}{radio resource management}
\acro{HetNet}{heterogeneous network}
\acro{CEPT}{European conference of postal and telecommunications administrations}
\acro{ANATEL}{agência nacional de telecomunicações}
\acro{RAT}{radio access technology}
\acro{RAN}{radio access network}
\acro{SC}{single connection}
\acro{FS}{fast switching}
\acro{DC}{dual connectivity}
\acro{C-RAN}{cloud \acl{RAN}}
\acro{NR}{new radio}
\acro{mmW}{millimeter wave}
\acro{KF}{Kalman filter}
\acro{KPI}{key performance indicator}
\acro{BB}{branch and bound}
\acro{BER}{bit error rate}
\acro{BS}{base station}
\acro{CDF}{cumulative distribution function}
\acro{CFN}{channel Frobenius norm}
\acro{DCP}{difference of convex functions program}
\acro{GTEL}{wireless telecommunications research group}
\acro{GUB}{global upper bound}
\acro{GLB}{global lower bound}
\acro{HMSINR}{highest minimum \ac{SINR}}
\acro{HPBW}{half power beamwidth}
\acro{HWCR}{highest \ac{WCR}}
\acro{HWSR}{highest weighted sum rate}
\acro{IBC}{interference broadcast channel}
\acro{IC}{interference channel}
\acro{KKT}{Karush-Kuhn-Tucker}
\acro{LB}{lower bound}
\acro{LiDAR}{light detection and ranging}
\acro{LOS}{line of sight}
\acro{LP}{linear programming}
\acro{MAP}{a maximum a posteriori probability}
\acro{MILP}{mixed integer linear programming}
\acro{mTRP}{multi transmission and reception point}
\acro{MIMO}{multiple input multiple output}
\acro{MINLP}{mixed integer nonlinear programming}
\acro{MIP}{mixed integer programming}
\acro{MISO}{multiple input single output}
\acro{ML}{machine learning}
\acro{MMSE}{minimum mean squared error}
\acro{MRC}{maximum ratio combining}
\acro{MRT}{maximum ratio transmission}
\acro{MSE}{mean squared error}
\acro{MU-MIMO}{multi-user multiple input multiple output}
\acro{MU}{multi-user}
\acro{NLP}{nonlinear programming}
\acro{NLOS}{non-line of sight}
\acro{NMSE}{normalized mean squared error}
\acro{NP}{non-polynomial time}
\acro{QoS}{quality of service}
\acro{SCA}{successive convex approximation}
\acro{SINR}{signal-to-interference-plus-noise ratio}
\acro{SISO}{single input single output}
\acro{SIMO}{single input multiple output}
\acro{SV}{steering vector}
\acro{TU}{typical urban}
\acro{UAV}{unmanned aerial vehicles}
\acro{UE}{user equipment}
\acro{UKF}{unscented Kalman filter}
\acro{JSFRA}{joint space-frequency resource allocation}
\acro{UB}{upper bound}
\acro{WINNER}{wireless world initiative new radio}
\acro{WMMSE}{weighted minimum mean squared error}
\acro{WCR}{weighted common rate}
\acro{V2V}{vehicle-to-vehicle}
\acro{V2I}{vehicle-to-infrastructure}
\acro{V2P}{vehicle-to-pedestrian}
\acro{V2X}{vehicle-to-everything}
\acro{RF}{radio frequency}
\acro{RL}{reinforcement learning}
\acro{RRV}{rate-relaxation variable}
\acro{DQN}{deep Q-network}
\acro{ADMM}{alternating direction method of multipliers}
\acro{SDP}{semidefinite programming}
\acro{PARAFAC}{parallel factor analysis}
\acro{V2N}{vehicle-to-network}
\acro{BALS}{bilinear alternating least-squares}
\acro{TALS}{trilinear alternating least-squares}
\acro{RMSE}{root mean square error}
\acro{ALS}{alternating least-squares}
\acro{5G-PPP}{5G infrastructure public private partnership}
\acro{SLNR-MAX}{maximum signal-to-leakage-plus-noise ratio}
\acro{ReLU}{rectifier linear unit}
\acro{ZF}{zero-forcing}
\acro{RUE}{resource usage efficiency}
\acro{5G-StoRM}{5G stochastic radio channel for dual mobility}
\acro{UMi}{urban micro}
\acro{DB}{database}
\acro{KNN}{K-nearest neighbors}
\acro{PCA}{principal component analysis}
\acro{CIR}{channel impulse response}
\acro{EE}{energy efficiency}
\acro{SE}{spectral efficiency}
\acro{LTE-A}{long term evolution advanced}
\acro{NOMA}{non-orthogonal multiple access}
\acro{FD}{full-duplex}
\acro{OFDMA}{orthogonal frequency-division multiple access}
\acro{PDPR}{pilot-to-data power ratio}
\acro{AWGN}{additive white Gaussian noise}
\acro{CFO}{carrier frequency offset}
\acro{SER}{symbol error rate}
\acro{M-QAM}{multilevel quadrature amplitude modulation}
\acro{B5G}{beyond fifth generation}
\acro{SSB}{synchronization signal block}
\acro{gNB}{gNodeB}
\acro{RSRP}{reference signal received power}
\acro{PSS}{primary synchronization signal}
\acro{SSS}{secondary synchronization signal}
\acro{PBCH}{physical broadcast channel}
\acro{PCI}{physical cell identity}
\acro{RACH}{random access channel}
\acro{RAR}{random access response}
\acro{BPSK}{binary phase shift keying}
\acro{TRP}{transmission reception point}
\acro{LP-MMSE}{local partial MMSE}
\acro{mMIMO}{massive multiple input multiple output}
\acro{GP}{geometric programming}
\acro{i.i.d.}{independent and identically distributed}
\acro{SIC}{successive interference cancellation}
\acro{LHS}{left hand side}
\acro{SOC}{second-order cone}
\acro{SOCP}{second-order cone program}
\acro{FD}{full-duplex}
\acro{SWIPT}{simultaneous wireless information and power transfer}
\acro{WPT}{wireless power transfer}
\acro{HD}{half-duplex}
\acro{SI}{self-interference}
\acro{IAI}{inter-AP interference}
\acro{IUI}{inter-user interference}
\acro{LMMSE}{linear minimum mean squared error}
\acro{InH}{indoor hotspot}
\acro{FPC}{fractional power control}
\acro{TS}{time switching}
\acro{6G}{sixth generation}
\acro{B5G}{Beyond fifth generation}
\acro{PS}{power splitting}
\acro{i.i.d}{independent and identically distributed}
\acro{CF}{cell-free}
\end{acronym}
\def\BibTeX{{\rm B\kern-.05em{\sc i\kern-.025em b}\kern-.08em
		T\kern-.1667em\lower.7ex\hbox{E}\kern-.125emX}}

\usepackage{tabularx}
\usepackage{supertabular} %
\usepackage{booktabs}

\DeclareMathAlphabet{\mathppl}{T1}{ppl}{m}{it}
\DeclareMathAlphabet{\mathphv}{T1}{phv}{m}{it}
\DeclareMathAlphabet{\mathpzc}{T1}{pzc}{m}{it}
\DeclareMathOperator{\re}{Re}

\newcommand{\Conj}[1]{{#1}^*}

\newcommand{\Herm}[1]{{#1}^{\mathrm{H}}}

\newcommand{\Mt}[1]{\mathbf{#1}}

\newcommand{\Real}[1]{\re\left\{#1\right\}}

\newcommand{\Set}[1]{\mathcal{\uppercase{#1}}}

\newcommand{\Transp}[1]{{#1}^{\mathrm{T}}}

\newcommand{\Vt}[1]{\mathbf{\lowercase{#1}}}
\newcommand{\mtC}{\Mt{C}}
\newcommand{\mtD}{\Mt{D}}

\newcommand{\mtF}{\Mt{F}}

\newcommand{\mtZero}{\Mt{0}}

\newcommand{\vtB}{\Vt{B}}

\newcommand{\vtN}{\Vt{N}}

\newcommand{\vtY}{\Vt{Y}}

\newcommand{\vtZero}{\Vt{0}}
\newcommand{\stC}{\Set{C}}

\newcommand{\stK}{\Set{K}}

\newcommand{\stM}{\Set{M}}
\newcommand{\stN}{\Set{N}}
\newcommand{\stO}{\Set{O}}
\newcommand{\stP}{\Set{P}}

\newcommand{\bbC}{\mathbb{C}}

\newcommand{\bbE}{\mathbb{E}}

\usepackage{enumitem}
\SetLabelAlign{myright}{\hss{$#1$}}
\newlist{nomenclature}{description}{1}
\setlist[nomenclature]{labelwidth=\widthof{\textbf{\small $\vtZero_{U}, \mtZero_{U \times B}$}}, labelindent=0pt, labelsep=1em, leftmargin=!, align=myright, font=\normalfont}

\newtheorem{prop}{Proposition}

\allowdisplaybreaks

\begin{document}

\title{Template for IEEE TVT \LaTeX\ Submission}
\title{Efficient Battery Usage in Wireless-Powered Cell-Free Systems with Self-Energy Recycling}

\author{
	Iran M. Braga Jr., Roberto P. Antonioli, G\'abor Fodor, \\ Yuri C. B. Silva and Walter C. Freitas Jr.
	\thanks{This work was supported in part by Ericsson Research, Technical Cooperation Contract UFC.48, in part by the Brazilian National Council for Scientific and Technological Development (CNPq), in part by FUNCAP, and in part by CAPES/PRINT Grant 88887.311965/2018-00.
	}
	\thanks{Iran M. Braga Jr., Roberto P. Antonioli, Yuri C. B. Silva and Walter C. Freitas Jr. are with the  Wireless Telecom Research Group (GTEL),
		Federal University of Cear\'a, Fortaleza 60455-760, Brazil. (e-mails: \{iran, antonioli, yuri, walter\}@gtel.ufc.br.). Roberto P. Antonioli is also with Instituto Atlântico, Fortaleza 60811-341, Brazil (e-mail: roberto\_antonioli@atlantico.com.br).}
	\thanks{G\'abor Fodor is with Ericsson Research, 16480 Stockholm, Sweden, and also with the Division of Decision and Control, KTH Royal Institute of Technology, 11428 Stockholm, Sweden (e-mail: gabor.fodor@ericsson.com)}
	\thanks{Copyright (c) 2015 IEEE. Personal use of this material is permitted. However, permission to use this material for any other purposes must be obtained from the IEEE by sending a request to pubs-permissions@ieee.org.}% <-this % stops a space
}

\markboth{IEEE Transactions on Vehicular Technology,~Vol.~XX, No.~XX, XXX~2015}
{}
%{Shell \MakeLowercase{\textit{et al.}}: Bare Demo of IEEEtran.cls for Journals}

\maketitle

\begin{abstract}
	This paper investigates wireless-powered cell-free {systems, in which} 
	the users send their uplink data signal while simultaneously harvesting energy from {network nodes and user terminals} 
	-- including {the transmitting user terminal itself} -- by performing self-energy recycling.
	In this {rather general setting}, a closed-form lower bound of the amount of harvested energy and 
	{the achieved signal-to-interference-plus-noise ratio} expressions are derived.
	Then, to improve the energy efficiency, we formulate the problem of minimizing the users' battery energy usage while satisfying minimum data rate requirements.
	Due to the non-convexity of the problem, a novel alternating optimization algorithm is proposed, and its proof of convergence is provided.
	Finally, numerical results show that the proposed method is more efficient than a state-of-art algorithm in terms of battery energy usage and outage rate.
\end{abstract}


\section{Introduction}
% Motivation
\ac{B5G} and \ac{6G} networks {are expected to} deal with a large number of devices, which are battery-limited. Thus, energy management plays a crucial role {in the current and next} generations of mobile networks.
In this context, 
\ac{WPT} has garnered several efforts in the industry and academia, and is considered a promising technology for \ac{B5G} and \ac{6G} networks.
Indeed, \ac{WPT} allows low-power devices, including rechargeable batteries and energy-constrained wireless sensors, 
to power up their battery by means of radio frequency signals.
In terms of enabling architectures, {\ac{CF} systems  offer} higher coverage probabilities 
than cell-based networks, {and provide surface-uniform signal strengths,} which reduces the impact of path losses experienced 
in wireless scenarios.  
{These advantageous features of \ac{CF} systems improve} the energy harvesting opportunities of cell-edge users \cite{Demir2021a}.
Specifically, the combination of \ac{SWIPT} and \ac{CF} systems is promising, 
since \ac{SWIPT} is attractive in terms of enabling to carry information and energy simultaneously \cite{Zhang2022}.
%{\it GF: A specific reference that refers to this "combination" would be nice.}

\ac{SWIPT} in \ac{CF} systems has been investigated in some works from the literature using the \ac{TS} protocol, such as in \cite{Demir2021a, Femenias2021, Zhang2022}.%, Zheng2021}.
Although the \ac{TS} protocol has low hardware complexity, it {suffers from some performance loss} in terms of \ac{SE}, 
because less resources are dedicated to information transmission.
In \cite{Galappaththige2021}, the \ac{PS} protocol was studied, which splits the received signal into {an energy-carrying} signal and information-detection signal.
In general, the \ac{PS} protocol {achieves higher} performance than the \ac{TS} protocol. %,since the \ac{TS} protocol can be treated as a special case of \ac{PS} protocol with binary split power ratios.
However, the \ac{PS} protocol receiver needs a radio frequency signal splitter, {which increases the hardware complexity}.
Hence, the \ac{PS} protocol is {often not viable} when dealing with constrained-size devices~\cite{Na2018}. 
%{\it GF: A reference here would be nice, maybe just re-citing one of the previous references.}

Also, \ac{SWIPT} and \ac{FD} can be combined to obtain advantages in terms of spectral and energy efficiency.
In contrast to conventional \ac{FD} systems, in which \ac{SI} is harmful, \ac{SI} can be beneficial 
in terms of energy source for harvesting energy by exploiting self-energy recycling. 
{Indeed, self-energy recycling can provide an extra energy source} for users 
in addition to the energy sent from the \acp{AP}. {For instance, in \cite{Nguyen2022} the authors proposed relay selection algorithms with self-energy recycling to improve outage probability in multi-relay \ac{FD} networks with non-identically distributed fading channels.
The authors in~\cite{Kwon2020} and~\cite{Shaikh2020} also studied self-energy recycling in the context of multi-relay \ac{FD} networks. However, 
as} far as we know, this is the first paper analyzing wireless-powered \ac{CF} systems 
with self-energy recycling. 

Specifically, the main contributions of this work are:
\begin{itemize}
	\item
	We derive a closed-form lower bound of the amount of harvested energy and the \ac{SINR} expressions 
	in wireless-powered \ac{CF} systems with self-energy recycling. 
	Differently from the \ac{TS} and \ac{PS} protocols, in the proposed method more resources are available 
	for data transmission and more energy is harvested due to self-energy recycling with low hardware complexity at the user-side;
	\item
	We investigate the problem of minimizing the users' battery energy usage with minimum data rate demands;
	\item
	Due to the non-convexity of the formulated optimization problem, 
	we propose a novel alternating optimization algorithm and provide its proof of convergence;
	and 
	\item
	{We analyze the performance} using simulations and compare the proposed solution with another state-of-art algorithm, showing that the proposed method is more efficient in terms of battery energy usage and outage rate.
\end{itemize}

%%%%%%%%%%%%%%%%%%%%%%%%%%%%%%%%%%%%%%%%%%%%%%%%%%%%%%%%%%%%%%%%%%%%%%%%%%%%%%%%%%%%%%%%%%%%%%%%%%%%%%%%%
{\textit{Notation:} % The following notation is used throughout this paper. %
We use uppercase and lowercase boldface to denote matrices and vectors, respectively. %
Plain letters are used for scalars. %
{$\Transp{(\cdot)}$, $\Herm{(\cdot)}$} and $\Conj{(\cdot)}$ denote the transpose, conjugate transpose and conjugate operations, respectively.
$\{x_i\}_{\forall i}$ denotes the set of elements $x_i$ for the values of $i$ denoted by the subscript expression.
$\mathbf{I}$ is the identity matrix.
$\Real{\cdot}$ takes the real value of a complex number.
$|\cdot|$ takes the absolute value of a complex number.
{The} expected value of a random variable is denoted by $\mathbb{E}[\cdot]$.
Other notational conventions are summarized in Table~\ref{TABLE:Notation}}: 

\begin{table}[t]
	\centering
	\caption{{Notation conventions.}}
	\label{TABLE:Notation}
	{\begin{tabular}{lll}
		\toprule
		\textbf{Notation} & \textbf{Description} \\
		\midrule
		$x_{m,k}$ & Variable related to the link between \ac{AP} $m$ and user $k$ \\
		\midrule
		$\ddot{x}_{k, j}$ & Variable related to the link between user $k$ and user $j$ \\
		\midrule
		$\dddot{x}_{m,q}$ & Variable related to the link between \ac{AP} $m$ and user $q$ \\
		\midrule
		$\bar{x}$ & Variable related to the \ac{LOS} component\\
		\midrule
		$x^{\text{p}}$ & Variable related to the pilot transmission\\	
		\midrule
		$x^{(\text{d})}$ & Variable related to the downlink transmission\\		
		\midrule
		$x^{(\text{u})}$ & Variable related to the uplink transmission\\		
		\bottomrule
	\end{tabular}}
\end{table}

\section{System Model}
\label{SEC:SYSTEM_MODEL}
We consider a \ac{CF} system consisting of  $M$ single-antenna \acp{AP} and $K$ single-antenna users.
We define $\stM$ and $\stK$ as the sets of \acp{AP} and users, respectively.
Furthermore, all \acp{AP} are connected via fronthaul links to a \ac{CPU}.
Each coherence interval is divided into two phases.
In the first phase, all users send their pilot sequences to the \acp{AP}, which estimate the channels to design precoding vectors for effective energy transfer and data reception in the second phase.
In the second phase, all users send their data signal while harvesting energy from the \acp{AP} and interfering users.
Let $\tau_{\text{c}}$ be the total number of samples per coherence interval.
Then, we have that $\tau_{\text{c}} = \tau_{\text{p}} + \tau_{\text{u}}$, where $\tau_{\text{p}}$ and $\tau_{\text{u}}$ are the length of the first and second phases, respectively.
Therefore, differently from the \ac{TS} and \ac{PS} protocols, in the proposed method more resources are available for energy harvesting and data transmission as they are both executed in the second phase,
%since there is not dedicated resources for energy harvesting,
thus more energy is harvested due to self-energy recycling. 
Also, no \ac{SI} canceling is {needed} at the {users as they do not need to decode the data signal while transmitting data on the uplink}, which reduces the hardware complexity at the user-side.
%{\it GF: No SI cancellation is performed or no SI cancellation is needed ? It is not clear at this point if only the AP is full-duplex or also the user terminals.}

{We denote by} $g_{m,k}\in\bbC$ the channel vector between user~$k$ and \ac{AP} $m$, 
by $\ddot{g}_{k, j}\in\bbC$ the \ac{IUI} channel between user $k$ and user $j$, 
and by $\dddot{g}_{m,q}\in\bbC$ the channel between \ac{AP} $m$ and \ac{AP} $q$. 
Note that $\dddot{g}_{m,m}$ is the \ac{SI} channel at \ac{AP} $m$ while $\dddot{g}_{m,q}$ for $m\neq q$ is referred to as the \ac{IAI} 
channel, since uplink signals received at \ac{AP} $m$ are corrupted by downlink signals sent from \ac{AP} $q$. 
Similarly, $\ddot{g}_{k,k}$ is the \ac{SI} channel at user $k$.
All channels are considered to be constant in each time-frequency
coherence interval and they are \ac{i.i.d} in different coherence intervals. In addition, all channels are modeled as spatially uncorrelated Rician fading channels with unknown phase shifts, i.e., each channel realization
can be expressed as
\begin{align}
	g_{m,k} &= \bar{h}_{m,k}e^{j\theta_{m,k}} + h_{m,k} \, ,\\
	\ddot{g}_{k,j} &= \ddot{\bar{h}}_{m,k}e^{j\theta_{k,j}} + {\ddot{h}}_{k,j} \, ,\\
	\dddot{g}_{m,q} &= \dddot{\bar{h}}_{m,q}e^{j\theta_{m,q}} + {\dddot{h}}_{m,q} \, ,
\end{align}
where $\bar{h}_{m,k}e^{j\theta_{m,k}}$, $\ddot{\bar{h}}_{k,j}e^{j\theta_{k,j}}$ and $\dddot{\bar{h}}_{m,q}e^{j\theta_{m,q}}$ denote the \ac{LOS} components and $h_{m,k}$, ${\ddot{h}}_{k,j}$ and ${\dddot{h}}_{m,q}$ correspond to the \ac{NLOS} small-scale fading with $h_{m,k}\sim\stC\stN({0},\beta_{m,k})$, ${\ddot{h}}_{k,j}\sim\stC\stN(0,\ddot{\beta}_{k,j})$ and ${\dddot{h}}_{m,q}\sim\stC\stN(0,\dddot{\beta}_{m,q})$, where $\beta_{m,k}$, $\ddot{\beta}_{k,j}$ and $\dddot{\beta}_{m,q}$ are the large-scale fading coefficients which account for path-loss and shadowing.
We assume that only $\left\{\bar{h}_{m,k}, \ddot{\bar{h}}_{k,j}, \dddot{\bar{h}}_{m,q}, \beta_{m,k}, \ddot{\beta}_{k,j}, \dddot{\beta}_{m,q}\right\}$ are known at the \ac{CPU}, such as in~\cite{Ozdogan2019, Demir2021a}.
Indeed, we consider a more realistic scenario where the phase shifts $\theta$ in the \ac{LOS} 
components are unknown, and assume that they are uniformly distributed in the interval $\left[0, 2\pi\right)$ \cite{Ozdogan2019, Demir2021a}.
Finally, we have that
\begin{align}
	w_{m,k} &\triangleq \bbE\left[g_{m,k}\Conj{g}_{m,k}\right] = \bar{h}_{m,k}\Conj{\bar{h}}_{m,k} + \beta_{m,k} \, , \\
	\ddot{w}_{k,j} &\triangleq\bbE\left[\ddot{g}_{k,j}\Conj{\ddot{g}}_{k,j}\right] = \ddot{\bar{h}}_{k,j}\Conj{\ddot{\bar{h}}}_{k,j} + \ddot{\beta}_{k,j} \, , \\
	\dddot{w}_{m,q} &\triangleq \bbE\left[\dddot{g}_{m,q}\Conj{\dddot{g}}_{m,q}\right] = \dddot{\bar{h}}_{m,q}\Conj{\dddot{\bar{h}}}_{m,q} + \dddot{\beta}_{m,q} \, .
\end{align}

%is {composed by} %modeled as the combination of
%the large-scale fading $\beta_{m,k}$ and the small-scale fading vector $\vtH_{m,k} \in \bbC^{N\times 1}$: $\vtG_{m,k} = \sqrt{\beta_{m,k}}\vtH_{m,k}$.
%Note that $\vtH_{m,k}\sim\stC\stN(0, \mtI_{N}),\ \forall(m,k)$ are \acl{i.i.d.} random variables~\cite{Ngo2017}.
%We assume that only the large scale fading coefficients, $\{\beta_{m,k}\}_{\forall(m,k)}$,
%are known at the \ac{CPU}, as they vary slowly and can be easily estimated. Indeed, the large scale coefficients are kept constant during $T$ coherence intervals or time slots.

\subsection{Channel Estimation in the First Phase}
We assume that $\tau_{\text{p}}$ mutually orthogonal pilot sequences {$\sqrt{\tau_{\text{p}}}\bm{\varphi} \in\bbC^{\tau_{\text{p}}\times 1}$ are used for channel estimation with $\lVert\bm{\varphi}\rVert^{2}~=~1$}. Thus, let $\sqrt{\tau_{\text{p}}\rho_{\text{p}}}\bm{\varphi}_{k}$ be the pilot sequence assigned to user $k$, for $k = 1,\ \ldots,\ K$, where $\rho_{\text{p}}$ is the transmit power of the pilot symbol, and $\stP_k$ the subset of users which are
assigned the same pilot sequence as user $k$, including itself. Then, the received pilot signal vector $\vtY_{m}^{\text{p}}\in\bbC^{1\times \tau_{\text{p}}}$ at the \ac{AP} $m$, in the first phase, is given by
\begin{align}
	\vtY_{m}^{\text{p}} &= \sum_{k=1}^{K} \sqrt{\tau_{\text{p}} \rho_{\text{p}}}g_{m,k}\Herm{\bm{\varphi}}_{k}  + \vtN_{m}^{\text{p}},
	\label{EQ:PILOT_SIGNAL_FIRST_PHASE}
\end{align}
where $\vtN_{m}^{\text{p}}\in\bbC^{1\times \tau_{\text{p}}}$
is the receiver noise with independent $\stC\stN(0,\sigma^2)$ entries, in which $\sigma^2$ is the noise power.

After obtaining the projection of $\vtY_{m}^{\text{p}}$ onto $\bm{\varphi}_{k}$,
given by 
$\check{y}_{m,k}^{\text{p}} = \vtY_{m}^{\text{p}} \bm{\varphi}_{k}$,
the phase-unaware \ac{LMMSE} estimate of $\{g_{m,k}\}_{\forall(m,k)}$ is~\cite{Ozdogan2019}
\begin{align}
	\hat{g}_{m,k} &= \sqrt{\tau_{\text{p}} \rho_{\text{p}}}w_{m,k}\psi_{m,k}^{-1}\check{y}_{m,k}^{\text{p}}, \label{EQ:ESTIMATED_CHANNEL}
\end{align}
where $\psi_{m,k} = \sum\limits_{j\in\stP_{k}}\tau_{\text{p}}\rho_{\text{p}}w_{m,j} + \sigma^{2}.$
%{\it GF: In the above formula, I would prefer to take the conjugate instead of Hermitian, since we are dealing with scalars.}
%{\it GF: We should check if $\mathds{E}(\check{y}_{m,k}^{\text{p}})=0$, otherwise the above formula is not correct, please
%see: \url{https://en.wikipedia.org/wiki/Complex\_random\_variable}, eq. (3).}
Note that the estimated channel $\hat{g}_{m,k}$ and the channel estimation error $\tilde{g}_{m.k} = g_{m,k} - \hat{g}_{m,k}$ are independent {variables} 
distributed as $\hat{g}_{m,k}\sim\stC\stN(0,\gamma_{m,k})$ and $\tilde{g}_{m,k}\sim\stC\stN(0,c_{m,k})$, where $\gamma_{m,k} = \tau_{\text{p}}\rho_{\text{p}}w_{m,k}^{2}\psi_{m,k}^{-1}$ and $c_{m,k} = w_{m,k} - \gamma_{m,k}$.

\subsection{Wireless Power Transfer in the Second Phase}
We consider a non-coherent energy harvesting, thus, the signal transmitted by \ac{AP} $m$ to user $k$ is given by
\begin{align}
	s_{m}^{\text{(d)}} = \sum_{k=1}^{K}\sqrt{p_{m,k}^{\text{(d)}}}\hat{g}_{m,k}s_{m,k},
\end{align}
where $s_{q,k}\in\bbC^{\tau_{\text{p}}\times 1}$ is an energy signal transmitted from
the $q$-th \ac{AP} to the $k$-th user, which are \ac{i.i.d} with $|s_{q,k}|^2 = 1$, and $p_{q,k}^{(\text{d})}$ is the power used by the $q$-th \ac{AP} to transmit to the $k$-th user.
The transmission power for each \ac{AP} should satisfy the maximum power limit which is $P_{m}^{\max}$ in the long-term, i.e.,
\begin{align}
	\bbE\left[\left|s_{m}^{\text{(d)}}\right|^2\right] = \sum_{k=1}^{K}p_{m,k}^{\text{(d)}}\gamma_{m,k}  \leq P_{m}^{\max}.
\end{align}
The received signal at user $k$ in the second phase is
\begin{align}
	z^{\text{(d)}}_{k} &= \sum_{m=1}^{M}\sum_{j=1}^{K}\sqrt{p_{m,j}^{\text{(d)}}}\Conj{g}_{m,k}\hat{g}_{m,j}s_{q,j} + \sum_{j=1}^{K}\sqrt{\eta_{j}}\ddot{g}_{k,j}x_{j} + n_{k}^{\text{E}},
\end{align}
where $n_{k}^{\text{E}}$ is the noise at the user $k$, which is neglected because the noise floor is too low for energy harvesting \cite{Demir2021a}.
{The second term is related to the signals coming from other users in the system. This occurs because, differently from~\cite{Demir2021a}, we also transmit data while harvesting energy at the user side.
Observe also that this term includes the user $k$ itself, which is the term related to the self-energy recycling that increases the harvested energy.}
The $\eta_j$ and $x_j$ terms are related to the uplink transmission and are defined in the next subsection.

{In this work, we consider a linear energy harvesting model. While such a model might be optimistic for practical applications, it allows us to obtain initial insights about cell-free systems with self-energy recycling. In fact, as shown~\cite{Demir2021a}, it is expected that the performance of cell-free systems when considering a non-linear energy harvesting model is lower if compared to a linear model.
	However, the study of cell-free systems with self-energy recycling and a non-linear energy harvesting model is left for future works.
	Therefore, assuming} that $\mu\in [0, 1]$ is the energy harvesting efficiency of the rectifier circuit \cite{Demir2021a},
% Thus,
the average harvested energy at user $k$ during the second phase is
\begin{align}
	E_{k}^{\text{(d)}} &= \tau_{\text{d}}\mu\bbE\left[\left|\sum_{m=1}^{M}\sum_{j=1}^{K}\sqrt{p_{m,j}^{\text{(d)}}} \Conj{g}_{m,k}\hat{g}_{m,j}s_{m,j}\right.\right.\nonumber\\
	&\hspace{20ex}\left.\left. +  \sum_{j=1}^{K}\sqrt{\eta_{j}}\ddot{g}_{k,j}x_j\right|^2\right] \, ,
\end{align}
{where the second term is related to the signals coming from other users in the system.}

\begin{prop}
	Assuming the phase-unaware \ac{LMMSE} channel estimator in \eqref{EQ:ESTIMATED_CHANNEL}, the average energy harvested for non-coherent downlink transmission is given in \eqref{EQ:AVERAGE_ENERGY_HARVESTED_SECOND_PHASE} at the top of the next page.
\end{prop}
\begin{proof}
	The proof follows similar steps as in \cite{Demir2021a}, and the detailed proof is omitted here due to the limitation of space.
\end{proof}
\begin{figure*}
	\begin{align}
		E_{k}^{\text{(d)}} &= \mu\tau_{\text{d}}\left(\sum_{m=1}^{M}\sum_{j=1}^{K}{p_{m,j}^{\text{(d)}}}\gamma_{m,j}w_{m,k} + \sum_{j=1}^{K}\eta_{j}\ddot{w}_{k,j} \right.\nonumber\\
		&\hspace{15ex}\left.+\left(\tau_{\text{p}}\rho_{\text{p}}\right)^2\sum_{m=1}^{M}\sum_{j\in\stP_{k}}p_{m,j}^{\text{(d)}}\left(2\beta_{m,k}\Real{\Conj{\bar{h}}_{m,k}\psi_{m,j}^{-1}w_{m,j}\bar{h}_{m,k}w_{m,j}\psi_{m,j}^{-1}} + \beta_{m,k}^{2}\left|\psi_{m,j}^{-1}w_{m,j}\right|^{2}\right)\right) \label{EQ:AVERAGE_ENERGY_HARVESTED_SECOND_PHASE}
	\end{align}
\end{figure*}

\subsection{Uplink Data Transmission in the Second Phase}
In this phase, we also assume that all $K$ users simultaneously send their data and receive energy from the \acp{AP}. Thus, the received signal at the $m$-th \ac{AP} $y_{m}^{(\text{u})} \in \bbC$ is modeled as
\begin{equation}
	y_{m}^{(\text{u})} = \sum_{k=1}^{K}  \sqrt{\eta_{k}}g_{m,k}x_{k} + \sum_{q=1}^{M}\sum_{k=1}^{K} \sqrt{p_{q,k}^{(\text{u})}\sigma_{m}^{\text{RSI}}}\dddot{g}_{m,q}\hat{g}_{q,k}s_{q,k} + n_{m}^{(\text{u})} \, ,
	\label{EQ:DATA_SIGNAL}
\end{equation}
where $x_{k}\in\bbC$, with $\bbE\{|x_{k}|^2\} = 1$, is the transmitted data symbol
by user $k$, $\sigma_{m}^{\text{RSI}}\in \left[0, 1\right)$ is the residual \ac{SI} and inter-\ac{AP} interference\footnote{Similar to \cite{Vu2019}, we assume that the residual \ac{SI} inter-\ac{AP} interference can be modeled in the same manner as for the \ac{SI}. Indeed, the \ac{IAI} can be seen as the virtual residual \ac{SI} of a virtual ``large" base station where all transmit and receive antennas are collocated due to the centralized processing in the \ac{CPU}. {Moreover, there are practical mechanisms that allow us to have residual \ac{SI} in the order of -110~dB\cite{Goyal2015,Bharadia2013,Nawaz2021,Ayati2021}, which practically enables \ac{FD} implementations.}} suppression level after all real-time cancellations in analog-digital domains~\cite{Nguyen2018a,Braga2021a}, $n_{m}^{(\text{u})}\sim\stC\stN(0,\sigma^{2})$ is the noise
on the received data signal, $\eta_{k}$ is the transmit power of the data symbol and $p_{q, k}^{\text{(u)}}$ is the power allocated by \ac{AP} $q$ to transmit to user $k$ in the second phase.
{The second term is related to the residual \ac{SI} at \ac{AP}~$m$.}

Also, each \ac{AP} is able to perform local data detection that are passed to the \ac{CPU} for final decoding. By employing the phase-unaware \ac{LMMSE} channel estimator in \eqref{EQ:ESTIMATED_CHANNEL} for \ac{MRC} {coherent} detection, the local estimate of $x_k$ at \ac{AP} $m$ is given by:
%
%{\it GF Since we assume MRC, we should clearly state that we consider coherence reception at the $M$ base stations ?}
%
\begin{align}
	\tilde{x}_{m,k} &= \sum_{j=1}^{K}  \sqrt{\eta_{j}}\Conj{\hat{g}}_{m,k}g_{m,j}x_{j}\nonumber\\
	&\hspace{2ex} + \sum_{q=1}^{M}\sum_{j=1}^{K} \sqrt{p_{q,j}^{(\text{p})}\sigma_{m}^{\text{RSI}}}\Conj{\hat{g}}_{m,k}{\dddot{g}}_{m,q}\hat{g}_{q,j}s_{q,j}+ n_{m}^{(\text{u})}.
\end{align}

Once the local estimates are sent to the \ac{CPU}, they are multiplied by a receive filter coefficient $\alpha_{m,k}$, i.e.,
\begin{equation}
	\hat{x}_{k} = \sum_{m\in\stM_{k}} \alpha_{m,k} \tilde{x}_{m,k}.
\end{equation}
\makeatletter
Note that the receive filter coefficients can be optimized to maximize the \ac{SE} using only channel statistics since the \ac{CPU} does not know the channel estimates.

In addition, let us define the following vectors and matrices for ease of
notation:
\begin{align*}
	&\bm{\alpha}_{k} \triangleq \Transp{\left[\alpha_{1,k}, \ldots, \alpha_{M,k}\right]}\in\bbC^{M\times 1} , \\
	&\vtB_{k} \triangleq \Transp{\left[b_{1,k}, \ldots, b_{M,k}\right]}\in\bbC^{M\times 1},\quad  b_{m, k} \triangleq \bbE\left[\Conj{\hat{g}}_{m,k}g_{m,k}\right] , \\
	&\mtC_{k,j}\in\bbC^{M\times M},\quad c_{k,j}^{m, m'}\triangleq \bbE\left[\Conj{\hat{g}}_{m,k}g_{m,j}\Conj{g}_{m',j}\hat{g}_{m',k}\right] ,\\
	&\mtD_{k} \in\bbC^{M\times M}, \quad d_{m,k} \triangleq \bbE\left[\Conj{\hat{g}}_{m,k}n_{m}^{\text{(u)}}\Conj{\left(n_{m}^{\text{(u)}}\right)}\hat{g}_{m,k}\right] ,\\
	&\mtF_{k,q,j} \in\bbC^{M\times M},\\
	&f_{k,j}^{m, m'} \triangleq \bbE\left[\Conj{\hat{g}}_{m,k}\dddot{g}_{m,q}\hat{g}_{q,j}\sqrt{\sigma_{m}^{\text{RSI}}}\sqrt{\sigma_{m'}^{\text{RSI}}}\Conj{\hat{g}}_{q,j}\Conj{\dddot{g}}_{m',q}\hat{g}_{m',k}\right],
\end{align*}
where $c_{k,j}^{m, m'}$ and $f_{k,j}^{m, m'}$ are the elements of the line $m$ and column $m'$ of the matrices $\mtC_{k,j}$ and $\mtF_{k,q,j}$, respectively, and $\mtD_{k}$ is a diagonal matrix with the $m$-th diagonal element being $d_{m,k}$. Thus, the \ac{SE} can be computed using proposition \ref{PROP:SE}.
\begin{prop}
	Assuming the phase-unaware \ac{LMMSE} channel estimator in \eqref{EQ:ESTIMATED_CHANNEL}, the closed-form expression for the achievable \ac{SE} of user $k$ using \ac{MRC} at the \acp{AP} when simultaneous data transmission and \ac{WPT} is performed in \ac{CF} systems is
	\begin{align}
		R_{k} = \tau_{\text{u}}/\tau_{\text{c}}\log_{2}(1 + \Gamma_{k}),
	\end{align}
	where $\Gamma_{k}$ is the effective \ac{SINR} shown in \eqref{EQ:EFFECTIVE_SINR} at the top of the next page. {The term $f_{k,q, j}^{m, m'}$ is not present in~\cite{Demir2021a} and arises in our analyses due to the residual \ac{SI} at the \acp{AP}.}
	\label{PROP:SE}
\end{prop}
\begin{proof}
	Based on \cite{Ozdogan2019, Braga2022a, Antonioli2022} where all additive interference is treated as a worst-case Gaussian noise we can arrive at the result. The detailed proof is omitted due to lack of space.
\end{proof}
\begin{figure*}
	\begin{align}
		&\Gamma_{k} = \frac{\eta_{k}\left|\Herm{\bm{\alpha}}_{k}\vtB_{k}\right|^{2}}{\Herm{\bm{\alpha}}_{k}\left(\sum\limits_{j=1}^{K}\eta_{j}\mtC_{k,j}\right)\bm{\alpha}_{k} - \eta_{k}\left|\Herm{\bm{\alpha}}_{k}\vtB_{k}\right|^2 + \Herm{\bm{\alpha}}_{k}\left(\sum\limits_{q=1}^{M}\sum\limits_{j=1}^{K}p_{q,j}^{\text{d}}\mtF_{k,q,j}\right)\bm{\alpha}_{k}+ \Herm{\bm{\alpha}}_{k}\mtD_{k}\bm{\alpha}},   \label{EQ:EFFECTIVE_SINR}\\
		&\mbox{where} \nonumber\\
		&b_{k}^{m} = \tau_{\text{d}}\rho_{\text{p}}\Conj{\bar{h}}_{m,k}\psi_{m,k}^{-1}w_{m,k}\bar{h}_{m,k} + \tau_{\text{d}}\rho_{\text{p}}\beta_{m,k}\psi_{m,k}^{-1}w_{m,k}, \qquad d_{k}^{m,m'} =\left\{\begin{matrix}
			\sigma^2\gamma_{m,k},&  m = m'\\
			0,& m\neq m'
		\end{matrix}\right.\nonumber\\% \label{EQ:ELEMENTS_OF_B}\\
		&c_{k,j}^{m,m'} = \left\{\begin{matrix}
			\gamma_{m,k}w_{m,j} + \tau_{\text{d}}^2\rho_{\text{p}}^2\left(2\beta_{m,j}\Real{\Conj{\bar{h}}_{m,j}\psi_{m,k}^{-1}w_{m,k}\bar{h}_{m,j}w_{m,k}\psi_{m,k}^{-1}} +\beta_{m,j}^2\left|w_{m,k}\psi_{m,k}^{-1}\right|^2\right)\left|\Herm{\bm{\phi}}_{j}\bm{\phi}_{k}\right|^2,& m = m' \\
			\tau_{\text{d}}^2\rho_{\text{p}}^2\left[\left(\Conj{\bar{h}}_{m,j}\psi_{m,k}^{-1}w_{m,k}\bar{h}_{m,j} + \beta_{m,j}\psi_{m,k}^{-1}w_{m,k}\right) \left(\Conj{\bar{h}}_{m',j}\psi_{m',k}^{-1}w_{m',k}\bar{h}_{m',j} + \beta_{m',j}\psi_{m',k}^{-1}w_{m',k}\right)\right]\left|\Herm{\bm{\phi}}_{j}\bm{\phi}_{k}\right|^2,& m \neq m'
		\end{matrix}\right. \nonumber\\%\label{EQ:ELEMENTS_OF_C}\\
		&f_{k,q, j}^{m, m'} =\left\{\begin{matrix}
			\dddot{w}_{m,q}\sigma_{m}^{\text{RSI}}\left[\tau_{\text{d}}^2\rho_{\text{p}}^2\left(2\beta_{q,j}\Real{\Conj{\bar{h}_{q,j}}w_{m,k}\psi_{m,k}^{-1}\bar{h}_{q,j}\psi_{m,k}^{-1}w_{m,k}} + \beta_{q,j}^2\left|w_{m,k}\psi_{m,k}^{-1}\right|^2\right) + \gamma_{m,k}c_{q,j}\right]\left|\Herm{\bm{\phi}}_{j}\bm{\phi}_{k}\right|^2 \\ +\dddot{w}_{m,q}\sigma_{m}^{\text{RSI}}\gamma_{m,k}\gamma_{q,j},&\hspace{-11ex} q = m \mbox{ and }  m = m'\\
			\dddot{w}_{m,q}\gamma_{m,k}\sigma_{m}^{\text{RSI}}\gamma_{q,j},&\hspace{-11ex} q \neq m \mbox{ and }  m = m'\\
			0,& m\neq m'
		\end{matrix}\right. \nonumber%\label{EQ:ELEMENTS_OF_F}\\
		\end{align}
\end{figure*}

\section{Problem Formulation}
\label{SEC:PROBLEM_FORMULATION}
In this paper we consider the problem of minimizing the users' battery energy usage while satisfying minimum data rate requirements. We assume that the users are powered over the air by the \acp{AP} and have a battery.
The battery energy is required to compensate for the randomness and intermittence of the harvested energy, increasing the feasibility of the problem (i.e., serving all users with their required minimum rate).
Then, considering $\eta_{k} = \eta_{k}^{\text{(e)}} + \eta_{k}^{\text{(b)}}$, where $\eta_{k}^{\text{(e)}}$ and $\eta_{k}^{\text{(b)}}$ denote the power drawn from the harvesting and the battery, respectively, the problem can be formulated as
%
%
%Hence, the consumed energy of user $k$ can be given by
%\begin{equation}
%	E_{\text{c}} = \tau_{\text{p}}p_{\text{p}} + \tau_{\text{d}}p_{\text{p}} + \tau_{\text{u}}\eta_{k}^{\text{(e)}} + \tau_{\text{u}}\eta_{k}^{\text{(b)}}
%\end{equation}
\begin{subequations}
	\begin{align}
		\underset{p_{m,k}^{\text{(d)}},\eta_{k}^{\text{(e)}},\eta_{k}^{\text{(b)}},\bm{\alpha_{k}}}{\text{minimize}}& \sum_{k}^{K} \tau_{\text{u}}\eta_{k}^{\text{(b)}} \label{EQ:PROBLEM_OBJ}&	\\
		\mbox{subject to } &R_{k} \geq R_{k}^{\text{th}}, \label{EQ:PROBLEM_CONS1}& \forall k, \\
		&\tau_{\text{u}}\eta_{k}^{\text{(e)}} \leq E_{k}^{\text{(d)}}, \; \tau_{\text{u}}\eta_{k}^{\text{(b)}} \leq E_{k}^{\max}, \label{EQ:PROBLEM_CONS2}&\forall k, \\
%		&\tau_{\text{u}}\eta_{k}^{\text{(b)}} \leq E_{k}^{\max}, \label{EQ:PROBLEM_CONS3}&\quad \forall k,\\
		&\sum_{k=1}^{K}p_{m,k}^{\text{(d)}}\gamma_{m,k} \leq P_{m}^{\max}, \label{EQ:PROBLEM_CONS4}& \forall m,		
	\end{align}
	\label{EQ:PROBLEM}
\end{subequations}
\hspace{-1ex}where $E_{k}^{\max}$ is the maximum available energy at the battery of user $k$ and $R_{k}^{\text{th}}$ is the minimum data rate demand of user $k$.

\section{Proposed Solution}
\label{SEC:PROPOSED_SOLUTION}
Note that problem \eqref{EQ:PROBLEM} is non-convex, which makes the global optimal solution computationally difficult to be obtained. Indeed, the non-convexity arises due to constraints \eqref{EQ:PROBLEM_CONS1} and the highly coupled variables. To deal with this issue, note that the data rate constraints \eqref{EQ:PROBLEM_CONS1} can be replaced by \ac{SINR} constraints as follows
\begin{equation}
	\Gamma_{k} \geq \Gamma_{k}^{\text{th}}, \quad \forall k, \label{EQ:SINR_CONSTRAINTS}
\end{equation}
where $ \Gamma_{k}^{\text{th}} = 2^{\frac{\tau_{\text{c}}R_{k}^{\text{th}}}{\tau_{\text{u}}}} - 1 $.
Moreover, to deal with the highly coupled variables, we can resort to the alternating optimization method, such as in \cite{Demir2021a}.
Note that $E_{k}^{\text{(d)}}$ as well as the numerator and denominator of \ac{SINR} expressions are linear in $\{p_{m,k}^{\text{(d)}}\}_{\forall m,k}$ and $\{\eta_{k}\}_{\forall k}$, given the receiver filter coefficient vector $\{\bm{\alpha}_{k}\}_{\forall k}$.
Hence, for some given $\bm{\alpha}_{k}$, the optimization problem can be shown to be quasilinear and its global optimum solution can be found using standard solvers such as CVX~\cite{cvx}.
In addition, the receiver filter coefficients $\bm{\alpha}_{k}$ only affect the \ac{SINR} of the user $k$ and can be found in closed form for the given uplink data power $\{\eta_{k}\}_{\forall k}$ by maximizing a generalized Rayleigh quotient~\cite{Demir2021a}.
The complete alternating optimization algorithm can be seen in Algorithm~\ref{ALG:PROPOSED}.
\begin{algorithm}[!thb]
	\small
	\caption{Proposed solution.}\label{ALG:PROPOSED}
	\begin{algorithmic}[1]
		\STATE Initialize $\{\bm{\alpha}_{k}\}_{\forall k}$ as all ones vector and set $l = 0$;
		\REPEAT
		\STATE Replace constraints \eqref{EQ:PROBLEM_CONS1} by \eqref{EQ:SINR_CONSTRAINTS};
		\STATE Solve problem \eqref{EQ:PROBLEM} by taking $\{\bm{\alpha}_{k}\}_{\forall k}$ as constants;
		\STATE Obtain the optimal $\{\bm{\alpha}_{k}\}_{\forall k}$ by maximizing a generalized Rayleigh quotient;
		\STATE Set $l \leftarrow l + 1$;
		\UNTIL Convergence has been reached or $l > L_{\max}$.
	\end{algorithmic}
\end{algorithm}

\begin{prop}
	The objective function of the original problem  \eqref{EQ:PROBLEM} converges by performing algorithm~\ref{ALG:PROPOSED}. \label{PROP:CONVERGENCE}
\end{prop}
\begin{proof}
	Note that the power values remain constant at each receiver filter coefficients update in line 5.
	Moreover, for fixed power values, the receiver filter coefficients optimization achieves an uplink data rate greater than or equal to that of the previous iteration.
	Then, less (or equal) transmission power will be needed at the next iteration.
	Hence, the objective function is monotonically decreasing with each iteration.
	Since the objective is bounded by the power and rate constraints, we can claim the convergence of~\eqref{EQ:PROBLEM_OBJ} with algorithm~\ref{ALG:PROPOSED}.
\end{proof}

Proposition~\ref{PROP:CONVERGENCE} shows the convergence of the objective function to a limit point of a monotonically decreasing sequence.
Since the original problem is non-convex, global optimality of the provided solution cannot be guaranteed.
Thus, the proof of optimality of the proposed method is an open research topic and left for future works.
Nevertheless, the simulation results in Section~\ref{SEC:NUMERICAL_RESULTS} demonstrate
that Algorithm~\ref{ALG:PROPOSED} provides significant performance gains
compared to the state-of-the-art solutions.

Furthermore, note that this method consists in a centralized
approach in which the CPU is responsible for all computations. Then, the CPU informs the values of $\{\eta_{k}^{(\text{(b)})}, \eta_{k}^{(\text{(e)})}\}_{\forall k}$ to the users. Finally, the per-iteration computational complexity of Algorithm~\ref{ALG:PROPOSED} is dictated by solving the linear problem in line 4, which has a complexity equivalent to $\stO\left((MK + 2K)^{3.5}\right)$.

\section{Numerical Results}
\label{SEC:NUMERICAL_RESULTS}

We consider the uplink of a \ac{CF} system with parameters $\{M, K, \tau_{\text{p}}, \tau_{\text{c}}\} = \{64, 4, 2, 200\}$.
We adopt the \ac{3GPP} \ac{InH} model in \cite{Demir2021a} with a 3.4 GHz carrier frequency and 20 MHz of bandwidth.
The large-scale fading coefficients, shadowing parameters, probability of LOS, and the Rician factors are simulated based on \cite{Demir2021a}.
%\footnote{The correlations between shadowing, terminal positions and Rician factors in \cite{3GPPTR36814} are neglected for simplicity. Moreover, the path loss of \ac{SI} channels are set to -15 dB as in \cite{Zeng2015}.}.
The \acp{AP} and users are uniformly distributed within a square of 100~m $\times$ 100~m and a 4~m height difference between \acp{AP} and users is taken into account when calculating distances.
A wrap-around technique is applied to imitate a network with an infinite area.
Moreover, each user randomly selects a pilot from a predefined set of orthogonal pilots.
We set $E_{k}^{\max} = 0.2$~J and $P_{m}^{\max}$ = 1~W.
The pilot power is fixed and set to 0.1~W.
The noise power is given by $\sigma^2 = -96$~dBm.
The energy harvesting efficiency of the rectifier circuit, $\mu$, is 0.5.
Finally, for comparison, we consider the solution proposed in \cite{Demir2021a} adapted to our problem in which the \ac{TS} protocol is adopted, i.e., it dedicates $\tau_{\text{d}}$ samples of the coherent interval for energy harvesting (thus, $\tau_{\text{c}} = \tau_{\text{p}} + \tau_{\text{d}} + \tau_{\text{u}}$).
We denote this solution as \ac{TS} in the plots.
Also, numerical results are obtained by 500 random realizations of \acp{AP} and users locations.

In Fig.~\ref{FIG:Convergence} we present the convergence of proposed algorithm with different values of $\mu$.
We can see that the proposed solutions converges to the final solution with a low number of iterations.
{The key insight taken from Fig.~\ref{FIG:Convergence} is that the} objective function improves as the energy harvesting efficiency of the rectifier circuit increases since more energy is being harvested, which decreases the energy usage from users' batteries.
Also, we observe that the number of iterations needed for convergence increases as $\mu$ increases.

\begin{figure}[!t]
	\centering
	\includegraphics[width=0.9\columnwidth]{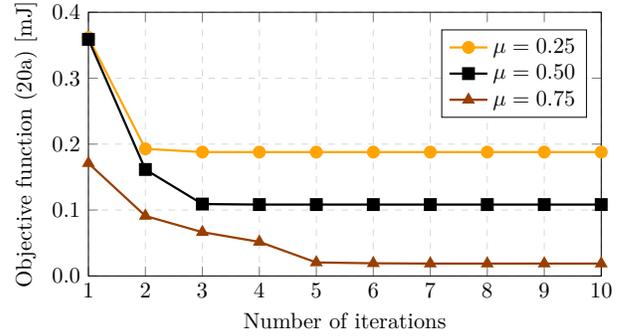}
	\caption{Convergence of proposed algorithm with $\sigma_{m}^{\text{RSI}} = -90$ dB and $R_{k}^{\text{th}} = 2.5$ bit/s/Hz.}
	\label{FIG:Convergence}
\end{figure}

Since the proposed solution allows users to transmit data while simultaneously receiving energy from the \acp{AP}, it is affected by the residual \ac{SI} at the \acp{AP}.
Thus, Fig.~\ref{FIG:RSI} presents the fraction of energy from the battery versus the residual \ac{SI}.
The fraction of energy from the battery is given by $\tau_{\text{u}}\eta_{k}^{\text{(b)}}/(\tau_{\text{u}}\eta_{k}^{\text{(b)}} + \tau_{\text{u}}\eta_{k}^{\text{(e)}})$.
As we can see, the \ac{TS} solution presents the same performance for all values of $\sigma_{m}^{\text{RSI}}$ since it is not affected by the \ac{SI}.
The proposed solution, in its turn, increases the fraction of energy used from the battery as the residual \ac{SI} increases.
The reason behind this is that to avoid \ac{SI}, the \acp{AP} allocate less power to the users, which affects the energy harvesting.
Thus, to fulfill the minimum data rates more energy is obtained from the battery.
Even so, {it is worth highlighting that} the proposed solution is able to outperform the \ac{TS} solution when the residual \ac{SI} is less than -95 dB{, which is a practical assumption considering existing \ac{SI} canceling mechanisms~\cite{Goyal2015,Bharadia2013,Nawaz2021,Ayati2021}, which allow us to have residual \ac{SI} in the order of -110~dB.}

\begin{figure}[!h]
	\centering
	\includegraphics[width=0.9\columnwidth]{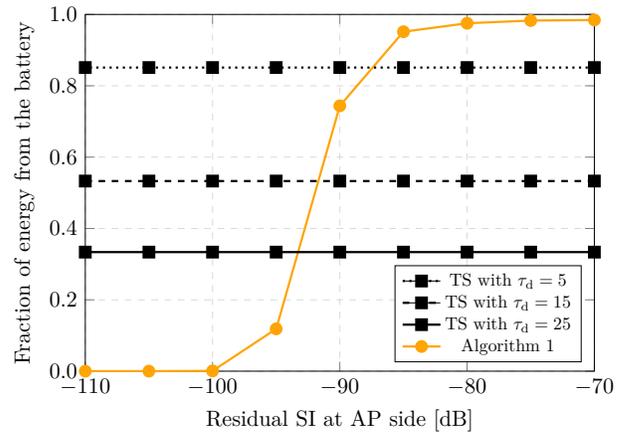}
	\caption{Fraction of energy from the battery versus residual SI ($\sigma_{m}^{\text{RSI}}$) with $R_{k}^{\text{th}} = 2.5$ bit/s/Hz.}
	\label{FIG:RSI}
\end{figure}

{Another aspect that should be highlighted is that while \ac{SI} energy-harvesting increases the harvested energy, the residual \ac{SI} affects the data decoding. 
To analyze this effect, Fig.~\ref{FIG:SE_vs_RSI} presents per-user \ac{SE} versus the residual \ac{SI} at the \acp{AP}. 
As it can be seen, the \ac{SE} increases as the residual \ac{SI} decreases. 
That is expected because low residual \acp{SI} yield higher data decoding rate at the \acp{AP}, leading to higher \ac{SE}. 
Moreover, the proposed method performs much better than the solution proposed in \cite{Demir2021a} in terms of \ac{SE} for values of residual \ac{SI} lower than -90 dB. 
This occurs because the baseline method harvests less energy, and - to minimize the battery consumption - it satisfies, but does not exceed, the data rate requirement. On the other hand, for lower and practical values of residual SI, the proposed method is able to harvest sufficient energy for data transmission without affecting the battery consumption. 
Since all harvested energy is used for data transmission, the user achieves data rates higher than the data requirements.}

\begin{figure}[!tb]
	\centering
	\includegraphics[width=0.9\columnwidth]{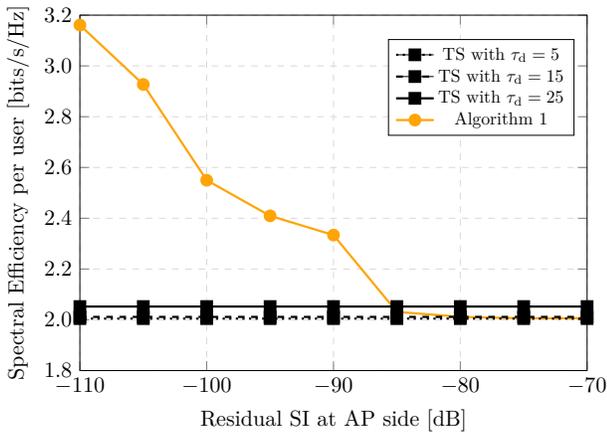}
	\caption{\ac{SE} per user versus the residual SI at the \acp{AP} ($\sigma_{m}^{\text{RSI}}$) with $R_{k}^{\text{th}}$ = 2.0 bits/s/Hz .}
	\label{FIG:SE_vs_RSI}
\end{figure}

Last but not least, in Fig.~\ref{FIG:OUTAGE} we present the outage rate versus the minimum data rate requirements.
We say that an outage event occurs when an algorithm is not able to find a feasible solution that fulfills all constrains of problem \eqref{EQ:PROBLEM} for a given realization of \acp{AP} and users' locations.
As expected, the outage rate increases as the minimum data rate requirements increases for all solutions.
Indeed, the higher the minimum data rate requirements, the harder it is to solve the problem.
The \ac{TS} solution degrades its performance as the number of dedicated samples for energy harvesting increases.
This occurs because less samples are dedicated to data transmission, which affects the \ac{SE} of the users.
Thus, the \ac{TS} solution performs a trade-off between minimizing the battery usage and fulfilling the rate demands.
{Meanwhile, we highlight that} the proposed solution achieves the best performance in terms of outage {with} no trade-off required.

\begin{figure}[!h]
	\centering
	\includegraphics[width=0.95\columnwidth]{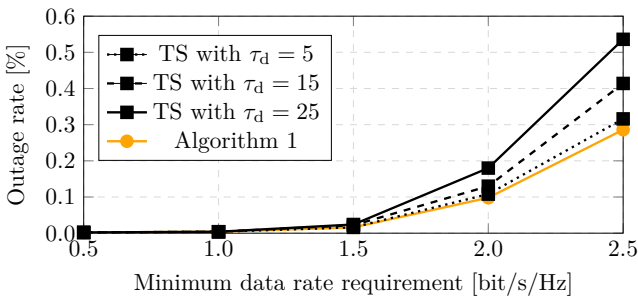}
	\caption{Outage rate versus minimum data rate requirements with $\sigma_{m}^{\text{RSI}} = -100$ dB and $R_{k}^{\text{th}} = 2.5$ bit/s/Hz.}
	\label{FIG:OUTAGE}
\end{figure}

\section{Conclusions}
\label{SEC:CONCLUSIONS}

In this paper, we studied wireless-powered \ac{CF} systems in which users transmit data signals to the \acp{AP} while simultaneously harvesting energy from the network nodes, including itself, by performing self-energy recycling.
{Due to this new feature of self-energy recycling, new terms arise in the derived equations, which have not been previously presented in the literature.}
%Closed-form lower bounds for the amount of harvested energy and \ac{SINR} expressions were derived, assuming uncorrelated Rician fading channel with unknown phase shifts and linear non-coherent energy harvesting.
The problem of minimizing the users' battery energy usage with minimum data rate constraints was formulated, for which an alternating optimization algorithm was proposed along with its proof of convergence.
Simulation results showed that the proposed method outperforms a state-of-art algorithm in terms of battery energy usage and outage rate.
Future works could consider coherent energy transmission and {non-linear energy harvesting}.

%\section*{Acknowledgment}
%This work was supported in part by Ericsson Research, Technical Cooperation Contract UFC.48, in part by the Brazilian National Council for Scientific and Technological Development (CNPq), in part by FUNCAP, in part by CAPES/PRINT Grant 88887.311965/2018-00 and in part by CAPES - Finance Code 001.

\bibliographystyle{IEEEtran}
\bibliography{IEEEabrv,cfrefs}

%\begin{thebibliography}{1}
%% You can use other form of bib file by changing here... 
%
%\bibitem{IEEEhowto:kopka}
%H.~Kopka and P.~W. Daly, \emph{A Guide to \LaTeX}, 3rd~ed.\hskip 1em plus
%  0.5em minus 0.4em\relax Harlow, England: Addison-Wesley, 1999.
%
%\end{thebibliography}

%\begin{IEEEbiography}{Yuguang ``Michael'' Fang}
%Biography text here.
%\end{IEEEbiography}
%
%%It is not necessary to upload the biography when you submit your manuscript.

\end{document}